%% file: timed-sandpiles.tex
\documentclass[runningheads]{llncs}

\input{preamble}

\begin{document}
\title{Timed Prediction Problem for Sandpile Models}
%
%
\author{Pablo Concha-Vega\inst{1} \and
Kévin Perrot\inst{1,2}}
\authorrunning{P. Concha-Vega et K. Perrot}
%
\institute{Aix Marseille Univ, CNRS, LIS, Marseille, France
\email{\{pablo.concha-vega,kevin.perrot\}@lis-lab.fr}
\and Université publique, Marseille, France
}
\maketitle              
\begin{abstract}
We investigate the computational complexity of the \emph{timed prediction problem}
in two-dimensional sandpile models. This question refines the classical
\emph{prediction problem}, which asks whether a cell $q$ will eventually become
unstable after adding a grain at cell $p$ from a given configuration.
The prediction problem has been shown
to be $\Poly$-complete in several settings,
including for subsets of the Moore neighborhood,
but its complexity for the von Neumann neighborhood remains open.
In a previous work, we provided a complete characterization of crossover gates
(a key to the implementation of non-planar monotone circuits)
for these small neighborhoods, leading to $\Poly$-completeness proofs
with only 4 and 5 neighbors among the height adjancent cells.
In this paper, we introduce the \emph{timed} setting, where the
goal is to determine whether cell $q$ becomes unstable \emph{exactly} at time $t$.
We distinguish several cases: some neighborhoods support complete timed toolkits
(including timed crossover gates) and exhibit $\Poly$-completeness; others admit
timed crossovers but suffer from synchronization issues; planar neighborhoods
provably do not admit any timed crossover; and finally, for some remaining
neighborhoods, we conjecture that no timed crossover is possible.

\end{abstract}
\section{Introduction}

Sandpiles are the canonical example of self-organized criticality~\cite{btw87},
a phenomenon where complex, scale-invariant behavior emerges from simple
local rules without any need for fine-tuning.
Originally proposed in the context of statistical physics, sandpile models have
also attracted interest from the perspective of computational complexity.
A central question in this context is the \textit{prediction problem}, which asks the
following: starting from a stable configuration, and choosing two specific cells
$p$ and $q$, is it possible that adding a single grain at cell $p$ will eventually
cause cell $q$ to become unstable (\emph{i.e.}, ready to topple) at some point in time?
This question captures the difficulty of anticipating the long-term effects of a
simple perturbation.

Early foundational work by Moore and Nilsson~\cite{mn99} revealed
that the computational power of sandpile dynamics depends strongly on the
dimension of the grid. In dimension one, they showed that the prediction problem
admits a fast parallel algorithm (in $\NC^3$), while in dimension three or more
it is $\Poly$-complete, implying an inherent sequentiality in the dynamics. The
two-dimensional case, however, remains open to this day.
A significant obstacle to proving $\Poly$-completeness in dimension two was identified
by Gajardo and Goles~\cite{gg06}, who showed that neither of the standard
von Neumann nor Moore neighborhoods admit a crossover gate,
which is a crucial component of all $\Poly$-completeness proofs so far.

In a previous work~\cite{cgmp2025}, we systematically studied all subsets
of the Moore neighborhood on the two-dimensional grid. By grouping them into
equivalence classes under rotation and reflection, we identified which classes admit
a \textit{crossover}, and which do not. It follows that the former have a $\Poly$-complete
prediction problem, whereas the latter may be computationally easier to predict.
In the present work, we extend this line of research by introducing a refined decision
problem, which we call the \textit{timed prediction problem}. This problem is similar
in spirit to the original prediction problem, but with an important difference: a
timestep $t$ is now part of the input. More precisely, given a stable configuration,
two cells $p$ and $q$, and a time $t$, the question is whether adding a grain at cell
$p$ will cause cell $q$ to become unstable after \emph{exactly} $t$ time steps.

We conduct a systematic study of the timed prediction problem across all subsets
of the Moore neighborhood, following the equivalence classes approach of~\cite{cgmp2025}.
For each class, we seek to determine whether it admits a complete set
of \textit{timed} computational gadgets.
Our results reveal a rich landscape: a large number of neighborhoods (planar ones) do
not admit a timed crossover at all, which makes it unlikely for them to support a
$\Poly$-hardness reduction. On the other hand, we identify several new classes that do admit
a complete timed toolkit, for which we establish the $\Poly$-completeness of the
timed prediction problem. A further group of neighborhoods admits a timed crossover
gate, but encounters a difficulty when connecting gates consistently in time, a
phenomenon we refer to as the \textit{delay issue}. For the remaining classes,
we conjecture that no timed crossover is possible.

\section{Definitions and state of the art}
\label{s:def_soa}

\subsubsection{Notations.}
We write $\N_+ = \N \setminus \{0\}$ for the set of positive integers,
and use $\fsubset$ to refer to finite subsets.  
Given any $d \in \N_+$, we denote by $0^d$ the $d$-dimensional tuple
consisting entirely of zeros.  
The Heaviside step function $\HH{\cdot}$ is defined by $\HH{\alpha} = 1$ 
if $\alpha \geq 0$, and $\HH{\alpha} = 0$ otherwise.  
For any $n \in \N_+$, we write $\int{n} = \{1,\dots,n\}$ and $\intz{n} = \{0,\dots,n-1\}$.

\subsubsection{Sandpile models.}

In this work, we consider sandpile dynamics restricted to two-dimensional square lattices $\Z^2$.  
A \emph{sandpile model} is specified by a finite \emph{neighborhood} $\neigh \fsubset \Z^2$,
whose cardinality we denote by $\theta_\neigh = |\neigh|$.  
A \emph{configuration} is a function $c: \Z^2 \to \N$ assigning to each cell $p \in \Z^2$
a number $c(p)$ of grains.  
The associated dynamics $F_\neigh: \N^{\Z^2} \to \N^{\Z^2}$ is defined cellwise by:

\[
  \forall p\in\Z^d: F_\neigh(c)(p)=
  c(p)
  -\theta_\neigh\cdot\HH{c(p)-\theta_\neigh}
  +\sum_{p'\in\neigh}\HH{c(p-p')-\theta_\neigh}.
\]

In words, each cell reaching the threshold $\theta_\neigh$ \emph{topples} or \emph{fires},
sending one grain to each of its out-neighbors,
whose relative positions are given by $\neigh$.
This local rule is applied in parallel at every cell of the lattice.
The most well-known neighborhoods, due to their simplicity and symmetries, are
\emph{von Neumann} and \emph{Moore}, respectively defined as:
\[
  \vn=\{(x,y)\in\Z^2 \mid |x|+|y|=1\}
  \quad\text{and}\quad
  \moore=\{(x,y)\in\Z^2 \mid 1\leq x^2+y^2\leq 2\}.
\]
See an example in Figure~\ref{fig:uspm}.

\begin{figure}[t]
  \centering
  \includegraphics[scale = .38, valign=c]{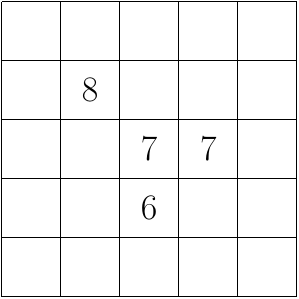}
  $\xrightarrow{F}$
  \includegraphics[scale = .38, valign=c]{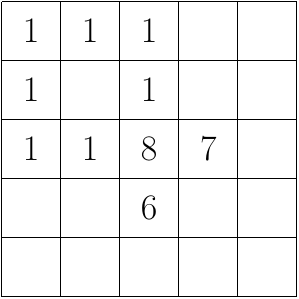}
  $\xrightarrow{F}$
  \includegraphics[scale = .38, valign=c]{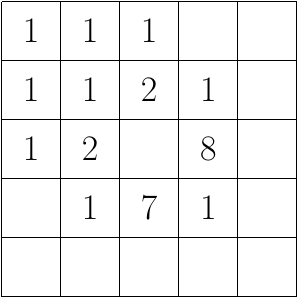}
  $\xrightarrow{F}$
  \includegraphics[scale = .38, valign=c]{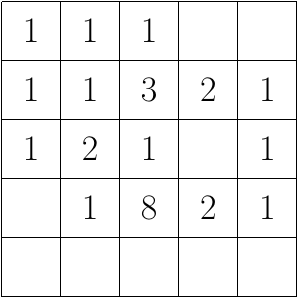}
  $\xrightarrow{F}$
  \includegraphics[scale = .38, valign=c]{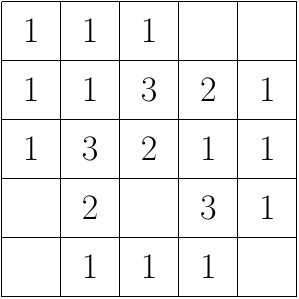}
  \caption{
    Example of sandpile dynamics for Moore neighborhood $\moore$. The numbers
    indicate the amount of grains in every cell (blank is empty).
  }
  \label{fig:uspm}
\end{figure}

When it is clear from the context, we drop the subscript notations
and simply denote $\theta$ and $F$ the threshold and the evolution rule.
The total number of sand grains in a configuration is denoted $\grains{c}=\sum_{p\in\Z^2} c(p)$,
it is invariant under $F$ for any sandpile model: $\grains{c}=\grains{F(c)}$.
A configuration $c$ is \emph{finite} when it contains a finite number of sand grains
($\grains{c}\in\N$), and \emph{stable} when no toppling occurs ($\forall p\in\Z^2:c(p)<\theta$).
From any finite configuration $c$, the dynamics converges to a stable configuration
denoted $c^\circ=\lim_{t\to\infty}F^t(c)$ (see \emph{e.g.},~\cite{fp19}).
Note that it is customary to exclude the cell itself (at relative position $0^d$) from its neighborhood,
in order to avoid ``trapping'' one sand grain at each cell.

We restrict our study to neighborhoods $\neigh\fsubset\Z^2$ that span $\Z^2$
\emph{i.e.}, such that $\neigh=\{p_1,p_2,\dots,p_k\}$ and 
  $\{ a_1\cdot p_1+a_2\cdot p_2\dots+a_k\cdot p_k \mid a_1,a_2,\dots,a_k\in\Z \} = \Z^2$.
Otherwise, the dynamics splits into independent subsystems, meaning that $\Z^2$
can be decomposed into sublattices for which the sand contents never interact.
In this case, the global behavior becomes a product of independent sandpile dynamics over
lower-dimensional components.
Such separability would complicate our considerations, as crossovers could trivially arise
between signals evolving in disconnected sublattices.

\subsubsection{Prediction problems and crossovers.}
The problems are formulated for each sandpile model $\neigh\fsubset\Z^2$.
The asymptotic prediction of sandpile dynamics is the problem of computing
the stable configuration $c^\circ$ reached from a given configuration $c$.
It is classically restricted to an input configuration $c$ which is finite and stable,
and a position $p\in\Z^2$ on which \emph{one} grain is added,
and formalized as the following decision problems~\cite{fp19}.
For any $p\in\Z^2$, let $\one_p$ be the configuration containing a unique grain at position $p$,
and let $+$ denote the cell-wise addition of two configurations
($\forall p\in\Z^2,(c+c')(p)=c(p)+c'(p)$).
It asks whether cell $q\in\Z^2$ topples during the evolution from $c+\one_p$ to $(c+\one_p)^\circ$.

\decisionpb
{$\neigh$-sandpile prediction problem}
{$\neigh$-PRED}
{a finite stable configuration $c\in\N^{\Z^2}$, and two positions $p,q\in\Z^d$.}
{does $\exists t\in\N,F^t(c+\one_p)(q)\geq\theta$ ?}

We highlight the following problem,
for which it is not clear yet whether it can be more difficult or less difficult than \textbf{$\neigh$-PRED}
for some sandpile models $\neigh$.
It asks whether cell $q$ topples at time step $t$ during the evolution from $c+\one_p$.

\decisionpb
{$\neigh$-sandpile timed prediction problem}
{$\neigh$-TIMED-PRED}
{a finite stable config.~$c\in\N^{\Z^2}$, two positions $p,q\in\Z^d$, and a time $t\in\N$.}
{does $F^t(c+\one_p)(q)\geq\theta$ ?}

The general goal is to dichotomize the complexity of prediction problems
based on the neighborhood $\neigh$:
either in $\NC$ or $\Poly$-complete for logspace reductions
(see~\cite{limits95,jaja92} for the definitions of $\NC$, $\Poly$,
and their meaning within the theory of computational complexity).
It is proven in~\cite{fp19} that all such problems lie in $\Poly$
(more precisely, the quantification over $t$ can be upper bounded to a value polynomial in the size of $c$).
On one hand, proofs that the problem belongs to $\NC$ typically involve either providing an efficient parallel algorithm,
or reducing the problem to another one for which such an algorithm is
known~\cite{fp19,gmmo17,gmt13,mn99}.
On the other hand, proofs of $\Poly$-completeness rely on the so-called \emph{Banks approach}~\cite{banks71},
which consists in reducing from \emph{monotone circuit value problem} (\textbf{MCVP})
by implementing the circuit computation via chains of reactions in the sandpile dynamics.
In this context, it is required to build \emph{wires}, \emph{turns},
\emph{and} gates, \emph{or} gates, \emph{diodes} and \emph{crossover} gates~\cite{gmpt17}.
In turns out that all gates are usually straightforward to design,
except crossover gates. The inexistence of crossover gates is therefore a strong indication
that such a reduction may not be possible, but it is relative to a precise definition of crossover gate
(to prove that none exist), which enforces a precise (hence restrictive)
consideration of what a \emph{wire}, or more generally a \emph{signal}, is.
This topic is of much broader interest~\cite{dm01}.

\begin{definition}[Crossover gate]\normalfont
A \emph{crossover gate} is a rectangular subset of the lattice
on which a stable configuration is set, such that it transports signals from two pairs of opposite sides,
independently from each other: two wires cross each other.
Given a sandpile model $\neigh$, a crossover gate of size $M\times N$ is formally defined as a configuration
$g:\intz{M}\times\intz{N}\to\{0,\dots,\theta-1\}$ together with four cells
$n,e,s,w\in(\intz{M}\times\intz{N})\setminus\{(0,0),(M-1,0),(0,N-1),(M-1,N-1)\}$ not in the corners,
verifying the following constraints~1--5.
For convenience, we denote the four sides of the gate as
$S_0=\{(0,y)\mid y\in\intz{N}\}$,
$S_1=\{(x,0)\mid x\in\intz{M}\}$,
$S_2=\{(M-1,y)\mid y\in\intz{N}\}$,
$S_3=\{(x,N-1)\mid x\in\intz{M}\}$.
\begin{enumerate}
  \item $n,s$ are on opposite sides, \emph{i.e.}
    $\exists i_n,n\in S_{i_n}$, $\exists i_s,s\in S_{i_s}$ and $|i_n-i_s|=2$;
  \item $w,e$ are on opposite sides, \emph{i.e.}
    $\exists i_w,w\in S_{i_w}$, $\exists i_e,e\in S_{i_e}$ and $|i_w-i_e|=2$;
  \item $n,e,s,w$ are on different sides, \emph{i.e.}
    $\{i_n,i_e,i_s,i_w\}=\{0,1,2,3\}$;
  \item adding a grain at $n$ eventually topples $s$, but no cell on the other sides, \emph{i.e.}
    $\exists t\in\N,F^t(g+1_n)(s)\geq\theta$,
    but $\forall p\in S_{i_w}\cup S_{i_e},\forall t\in\N,F^t(g+1_n)(p)<\theta$;
  \item adding a grain at $w$ eventually topples $e$, but no cell on the other sides, \emph{i.e.}
    $\exists t\in\N,F^t(g+1_w)(e)\geq\theta$,
    but $\forall p\in S_{i_n}\cup S_{i_s},\forall t\in\N,F^t(g+1_w)(p)<\theta$.
\end{enumerate}
\end{definition}

%
It is necessary to provide this formal definition, as we will state some impossibility results;
however it is difficult to use in practice.
In the following, no precise reference will be made to its constituents.
The main tool we will employ is the \emph{pair of firing graphs},
which is introduced at the beginning of Section~\ref{s:timed_crossover}.

For the timed prediction problem, a \emph{timed crossover gate} is required.
It is the same as a crossover gate,
except that all quantifications over time $t$ are bounded to a given value $T\in\N^+$
called the \emph{delay} of the gate
(which is part of the definition of the gate). That is,
when adding a grain to one side (at $n$ or $w$), the cell on the opposite side ($s$ or $e$)
must topple at time step $T$, and no cell on the other sides must topple before time step $T+1$.

\section{Timed crossover among subsets of Moore}
\label{s:timed_crossover}

In this section we present an exhaustive study of the computational complexity of sandpile dynamics,
by considering timed prediction problems (\textbf{$\neigh$-TIMED-PRED})
and timed crossover gates for every subset of the Moore neighborhood $\moore$.
Observe that some of the $2^8$ subsets of $\moore$ (which has $8$ elements) do not span $\Z^2$,
and the existence of a crossover gate is invariant by rotation and axial symmetries of the neighborhood.
It reduces the number of neighborhoods to study to $43$ equivalence classes,
presented in~\cite[Table 1]{cgmp2025}.

An example of timed crossover gate is presented on Figure~\ref{fig:timedcrossovergate}.
We will show that while some neighborhoods do not support a crossover gate, they do admit a
timed crossover gate and exhibit a $\Poly$-complete timed prediction problem.
It hints at the fact that predicting the dyanmics at a precise time step
may be harder than asymptotic prediction of the dynamics (the final stable configuration).

\begin{definition}[Subset of Moore neighborhood]
  Let \texttt{n},\texttt{ne},\texttt{e},\texttt{se},\texttt{s},\texttt{sw},\texttt{w},\texttt{nw}
  denote the eight cardinal coordinates, starting with \emph{north} and proceeding clockwise.
  We denote a subset of Moore neighborhood as an 8-bits string,
  where the $i$-th bit indicates whether the $i$-th coordinate in the list 
  \texttt{n},\texttt{e},\texttt{s},\texttt{w},\texttt{nw},\texttt{ne},\texttt{se},\texttt{sw},
  is part of the neighborhood (bit \texttt{1}), or not (bit \texttt{0}).
  We also interpret 8-bits strings as numbers (converting them to decimal)
  with the least significant bit on the right,
  see Figure~\ref{fig:moorestr}.
\end{definition}

\begin{figure}[t]
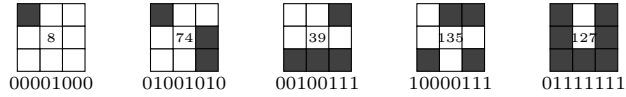

  \centerline{
    \mooregridcenternumber{00001000}
    \mooregridcenternumber{01001010}
    \mooregridcenternumber{00100111}
    \mooregridcenternumber{10000111}
    \mooregridcenternumber{01111111}
  }
  \caption{Subsets of Moore neighborhood corresponding to
  $\texttt{00001000}=8$,
  $\texttt{01001010}=74$,
  $\texttt{00100111}=39$,
  $\texttt{10000111}=135$,
  $\texttt{01111111}=127$.
  }
  \label{fig:moorestr}
\end{figure}

In this numbering scheme, the von Neumann neighborhood ($\vn$) corresponds to $\texttt{11110000}$
(decimal $240$),
while the Moore neighborhood ($\moore$) corresponds to $\texttt{11111111}$ (decimal $255$).
Note that neighborhoods with numbers $\leq 15$ consist exclusively of diagonal
neighbors, and as a result, they do not span $\Z^2$.


\begin{figure}[t]
  \centering
  \includegraphics[scale=.36]{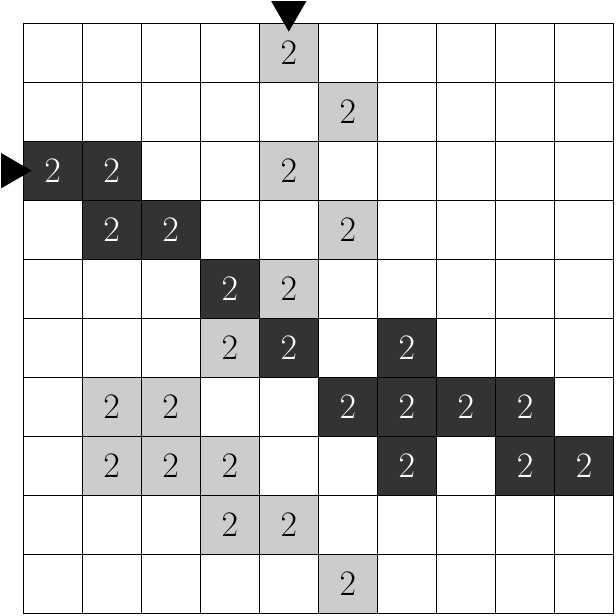}\quad
  \includegraphics[scale=.36]{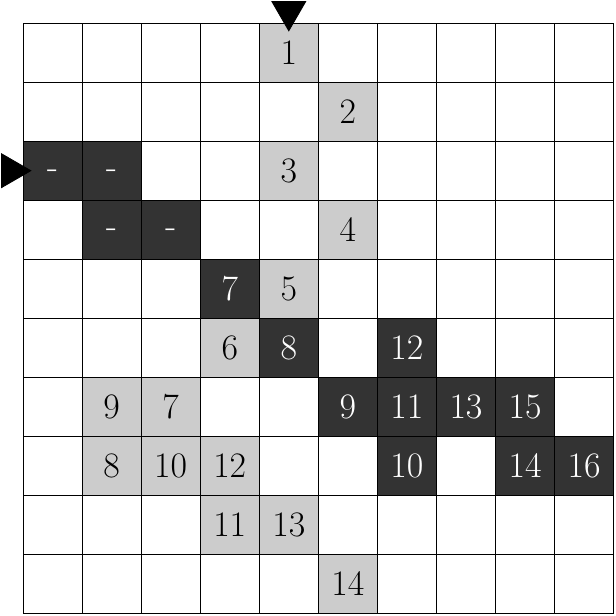}\quad
  \includegraphics[scale=.36]{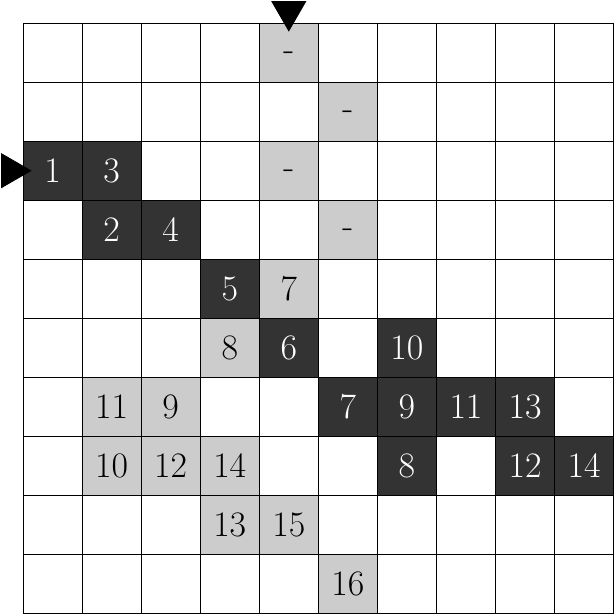}
  \caption{
    Example of timed crossover gate of size $10 \times 10$ for neighborhood $131$ (left).
    The delay of the gate is $14$.
    We picture the timestamps of each cell,
    following a grain addition at $n$ (center),
    and a grain addition at $w$ (right).
  }
  \label{fig:timedcrossovergate}
\end{figure}

\subsubsection{Section outline.}
We start this section by defining timed firing graphs (Subsection~\ref{ss:firingraphs-T}),
a tool that captures the timing and causal dependencies of topplings in a timed crossover.
First we prove that timed crossover gates are impossible when the neighborhood is planar (Subsection~\ref{ss:planar-T}).
Then we exhibit timed crossover gates for many neighborhoods.
For some of them we prove that \textbf{$\neigh$-TIMED-PRED} is $\Poly$-complete (Subsection~\ref{ss:crossover-T}),
and for others we are faced with a delay issue when trying to connect the different gates (Subsection~\ref{ss:crossover-T-delay}).
The results of this section are summarized in Table~\ref{tab:res_cross-T}.

\begin{table}[t]
  \centering
  \scalebox{.8}{
  	\input{results_table}
  }
  \caption{Summary of results exposed in Section~\ref{s:timed_crossover}.}
  \label{tab:res_cross-T}
\end{table}

\subsection{Timed firing graphs}
\label{ss:firingraphs-T}

The causal structure of topplings is captured by the concept of firing graph, introduced in~\cite{gg06}.
It is at the basis of our reasonings in this work.
Given a neighborhood $\neigh$, we say that cell $v_1$ is an \emph{in-neighbor}
(resp.~\emph{out-neighbor}) of cell $v_2$
when $(v_2-v_1)\in\neigh$ (resp.~$(v_1-v_2)\in\neigh$).
Observe that $v_1$ may be both an in-neighbor and an out-neighbor of $v_2$.
In particular, this is always the case for symmetric neighborhoods,
such as von Neumann $\vn$ and Moore $\moore$.

\begin{definition}[Firing graphs]\label{def:fg}\normalfont
  Given a crossover gate $g$ on cells $n,s,w,e$ 
  we define its two \emph{firing graphs} $G_{ns}^g = (V_{ns}, A_{ns})$ and $G_{we}^g = (V_{we}, A_{we})$ as:
  \begin{itemize}
    \item $V_{ns}$ (resp.~$V_{we}$) is the set cells toppled after adding a grain at cell $n$ (resp.~$w$);
    \item we have $(v_1,v_2) \in A_{ns}$ (resp.~$A_{we}$) 
      when $v_1,v_2 \in V_{ns}$ (resp.~$V_{we}$) verify:
      \begin{itemize}
        \item $v_1$ is an in-neighbor of $v_2$ for $\neigh$, and
        \item $v_1$ is fired strictly before $v_2$ after adding a grain at cell $n$ (resp.~$w$).
    \end{itemize}
  \end{itemize}
\end{definition}

The endpoints of a firing graph are called \emph{starting cells} ($n$ and $w$)
and \emph{ending cells} ($s$ and $e$) respectively.
For every arc $(v_1, v_2)$ in a firing graph, we call $v_1$ a \emph{predecessor} of $v_2$,
and $v_2$ a \emph{successor} of $v_1$.
The transitive and reflexive closure of the predecessor (resp.~successor)
relationship is the \emph{ancestor} (resp.~\emph{descendant}) relationship.
By definition, a predecessor is necessarily an in-neighbor, and a successor is necessarily an out-neighbor.
The firing graph captures the \emph{effective} causal relationships among cells
within a given crossover gate, whereas the in-neighbor and out-neighbor relationships represent \emph{potential} relationships.

The pair of firing graphs of a timed crossover gate may intersect (have vertices in common),
because the wire in the perpendicular direction can be toppled
(provided that it does so some steps behind compared to the $1$ signal traversing the gate).
We introduce the notion of \emph{timed firing graphs}, which are subsets of the firing graphs.
Given that the arcs of the firing graphs represent the causal relationship among toppled cells,
it would be enough to discard the part not necessary to topple an ending cell,
by restricting the corresponding firing graph to its ancestors.
However, this inductive definition from the ending cell backwards,
is not convenient for proofs by induction on the number of steps.
We prefer the following definition, based on timestamps,
which we denote $t_{ns}:V_{ns}\to\N$ and $t_{we}:V_{we}\to\N$
to distinguish the two firing graphs.
In each firing graph, we take into account only cells toppling sufficiently early to belong to $1$ signals.
Recall that the timestamps of the two ending cells are equal to the delay $T$ of the gate.
Let $t_{ns}(v) = +\infty$ when $t_{ns}$ is not defined for $v$.
The same applies to $t_{we}$.

\begin{definition}[Timed firing graphs]
  Given a timed crossover gate $g$ of delay $T\in\N_+$ on cells $n,s,w,e$, 
  its two \emph{timed firing graphs} are the subsets of its firing graphs,
  restricted to vertex sets:
  \[
    V^{T}_{ns}=\{v\in V_{ns}\mid t_{ns}(v) \leq t_{we}(v)\}
    \text{ and }
    V^{T}_{we}=\{v\in V_{we}\mid t_{we}(v) \leq t_{ne}(v)\}.
  \]
\end{definition}

It is necessary that in a timed crossover gate, each of the two ending cells belongs only
to its corresponding timed firing graph.
The $\leq$ comparison among timestamps aims at discarding cells that have some retard.
The following Lemma generalizes Proposition~1 from~\cite{np18},
to the context of timed firing graphs.

\begin{lemma}\label{lemma:distinct-T}
  If a sandpile model $\neigh$ admits a timed crossover gate,
  then it also admits a timed crossover gate $g$ with timed firing graphs
  $G^T_{ns} = (V^T_{ns}, A^T_{ns})$ and $G^T_{we} = (V^T_{we}, A^T_{we})$,
  such that $V^T_{ns} \cap V^T_{we} = \emptyset$.
\end{lemma}

\begin{proof}
  The proof is constructive:
  given a timed crossover gate $g'$ with timed firing graphs
  $G'^T_{ns} = (V'^T_{ns}, A'^T_{ns})$ and $G'^T_{we} = (V'^T_{we}, A'^T_{we})$,
  we:
  \begin{enumerate}
    \item remove all grains from vertices in the intersection $V'^T_{ns}\cap V'^T_{we}$
      (which are intuitively useless), so that they do not topple anymore;
    \item compensate for the missing predecessors in each timed firing graph:
      for each vertex in exactly one of $V'^T_{ns}$ or $V'^T_{we}$,
      add as many grains to it as it had predecessors in $V'^T_{ns}\cap V'^T_{we}$.
  \end{enumerate}
  We now argue that it gives another timed crossover gate with distinct timed firing graphs.
  By induction on successive time steps, we have the following three simple facts.
  First, the vertices in $V'^T_{ns}\cap V'^T_{we}$ do not topple anymore
  (this uses the fact that sandpile graphs are Eulerian).
  Second, the vertices not in the intersection are still in their timed firing graph
  (thanks to the sand grains added at step 2, there is no retard in the topplings).
  Third, no new vertex appears in each timed firing graph,
  \emph{i.e.} $V^T_{ns} \subseteq V'^T_{ns}$ and $V^T_{we} \subseteq V'^T_{we}$
  (because the sand grains added at step 2 do not create successor relationships in the timed firing graph).
  Given that $n,s\in V'^T_{ns}\setminus V'^T_{we}$
  and $w,e\in V'^T_{we}\cap V'^T_{ns}$,
  this still holds and the constructed configuration
  is indeed a timed crossover gate
  (it still performs a timed crossover of two signals),
  with sets of vertices $V^T_{ns}=V'^T_{ns}\setminus V'^T_{we}$
  and $V^T_{we}=V'^T_{we}\setminus V'^T_{ns}$.
  \qed
\end{proof}

\subsection{Timed crossover impossibility: planar neighborhoods}
\label{ss:planar-T}

As a corollary of Lemma~\ref{lemma:distinct-T},
it is impossible to have a timed crossover gate
when the neighborhood is planar.

\begin{corollary}
  If neighborhood $\neigh$ has a planar sandpile graph,
  then it does not admit a timed crossover gate.
\end{corollary}

This proves that neighborhoods
$66$, $74$, $82$, $98$, $90$, $106$, $130$, $146$,
$192$, $200$, $202$, $208$, $210$, $226$, $234$, $240$, $242$
and $250$ do not admit a timed crossover gate.

\subsection{Timed crossover possibility: $\Poly$-complete neighborhoods}
\label{ss:crossover-T}

Neighborhoods studied in this subsection:

\begin{center}
  \mooregridcenternumber{01101111} \hspace*{-.6cm}
  \mooregridcenternumber{01111111} \hspace*{-.6cm}
  \mooregridcenternumber{10000111} \hspace*{-.6cm}
  \mooregridcenternumber{10001111} \hspace*{-.6cm}
  \mooregridcenternumber{10010111} \hspace*{-.6cm}
  \mooregridcenternumber{11000011} \hspace*{-.6cm}
  \mooregridcenternumber{11000111} \hspace*{-.6cm}
  \mooregridcenternumber{11010011} \hspace*{-.6cm}
  \mooregridcenternumber{11010111}
\end{center}

We now prove that some subsets of Moore neighborhood have a $\Poly$-complete
\textbf{$\neigh$-TIMED-PRED} problem.
The proof is analogous to Theorem~2 from~\cite{cgmp2025},
but taking into account the delay of the gates.

\begin{theorem}
  \textbf{$\neigh$-TIMED-PRED} is $\Poly$-complete for neighborhoods
  $111$, $127$, $135$, $143$, $151$, $195$, $199$, $211$, $215$.
\end{theorem}

\begin{proof}
  For neighborhoods $135$ and $143$, this is a direct consequence
  of~\cite[Theorem 2]{cgmp2025}, as they both have a $\Poly$-complete
  prediction problem.

\begin{figure}[t]
  \centerline{\includegraphics[scale=.8]{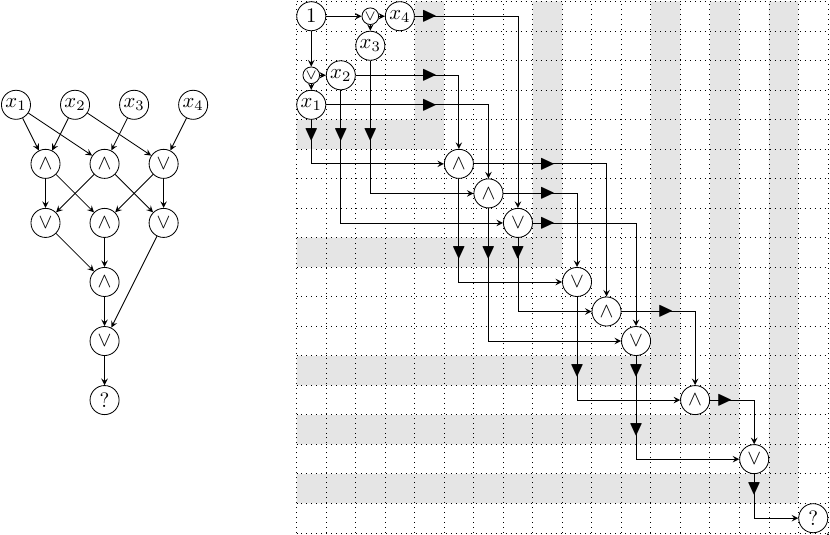}}
  \caption{
    Embedding of a \textbf{MCVP} instance of fan-in and fan-out two (left) on the grid (right).
    Diodes (represented as $\blacktriangleright$ and $\blacktriangledown$) are used systematically between layers
    in order to enforce the flow of information (and prevent it to go backward).
  }
  \label{fig:cvp-grid}
\end{figure}

  For the other neighborhoods, we prove $\Poly$-completeness by reduction
  from monotone circuit value problem (\textbf{MCVP}) with gates of fan-in
  and fan-out two (problem \textbf{AM2CVP} in~\cite{limits95}),
  We consider an embedding of the monotone circuit
  on the two-dimensional grid (see Figure~\ref{fig:cvp-grid}),
  and replace all the elements by the corresponding sandpile gadgets of fixed size.
  For each neighborhood, we present a set of gates having a uniform delay
  (it takes the same number of steps for a signal to traverse each gate).
  Let us denote $d\in\N_+$ this \emph{delay},
  and $z$ the peculiar delay of the starting constant $1$ gate
  (time necessary for the single grain addition at $p$ to reach the borders of the constant $1$ gate).
  Following Banks approach, a signal $1$ is an avalanche (chain of topplings)
  reaching the edges of gates at time steps $(z+d\cdot T)_{T\in\N}$,
  and a signal $0$ is a steady state or an avalanche
  reaching the edges of gates on time steps $z+d\cdot T+i$ with $i>0$
  (retards may accumulate).
  One can check on the layout of the \textbf{MCVP} instance on the grid (Figure~\ref{fig:cvp-grid}),
  that following a single grain addition, all the variable signals start synchronized
  (with avalanches for $1$ signals, and steady wires for $0$ signals).
  The uniform delay on all gates ensures that signals
  keep synchronized, because gates are placed on the diagonal:
  the two input signals are at the same distance from the origin.
  Every wire follows a shortest path in the grid toward the questioned cell.
  The latter therefore topples at time step $z+d\cdot U+y$,
  with $U$ the distance in the grid from the origin (top left cell) to the questioned cell
  and $y$ for its inner mechanism (precise distance to $q$),
  if and ony if the \textbf{MCVP} outputs a signal $1$.

\begin{figure}
  \centering
  \includegraphics[scale=.35]{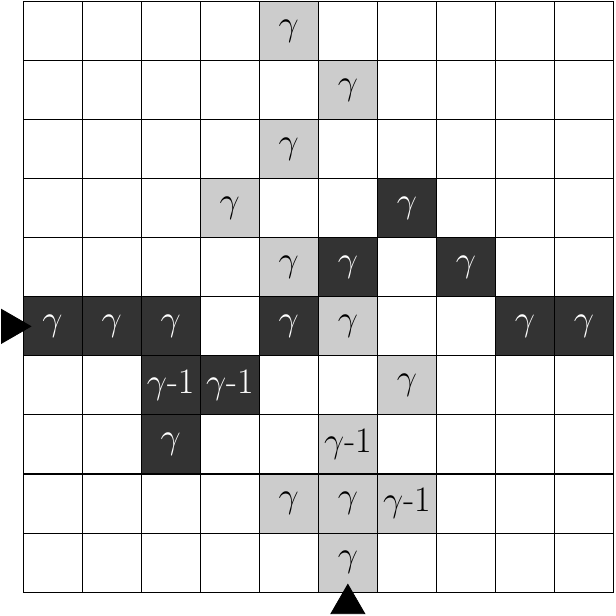}
  \hspace*{.5cm}
  \includegraphics[scale=.35]{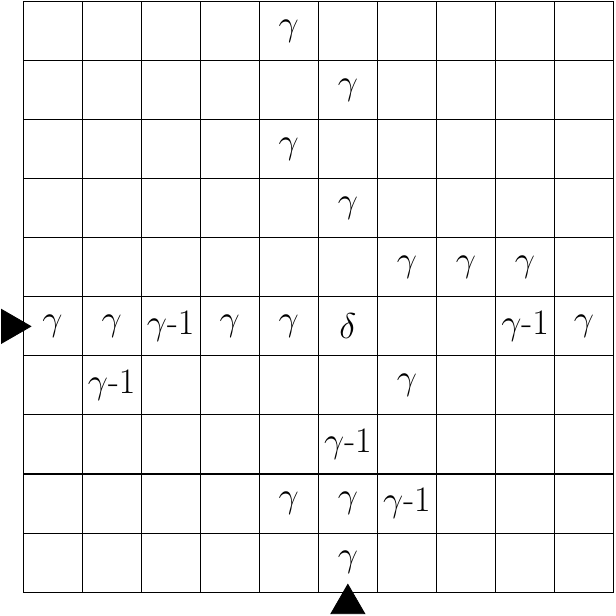}
  \\[.5em]
  \includegraphics[scale=.35]{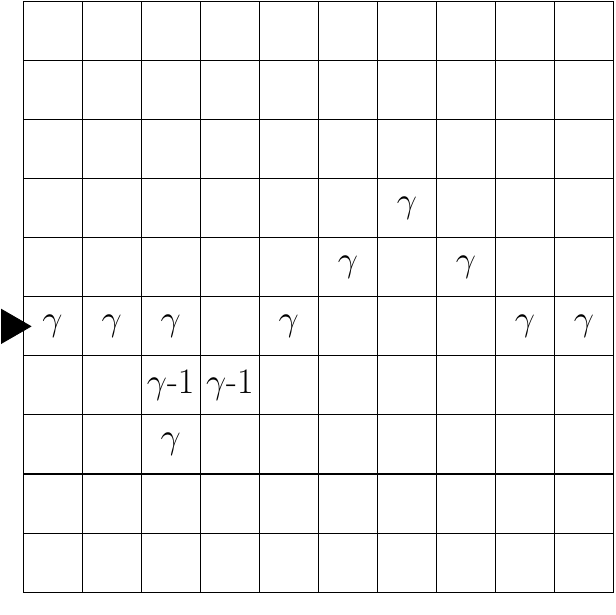}
  \hspace*{.5cm}
  \includegraphics[scale=.35]{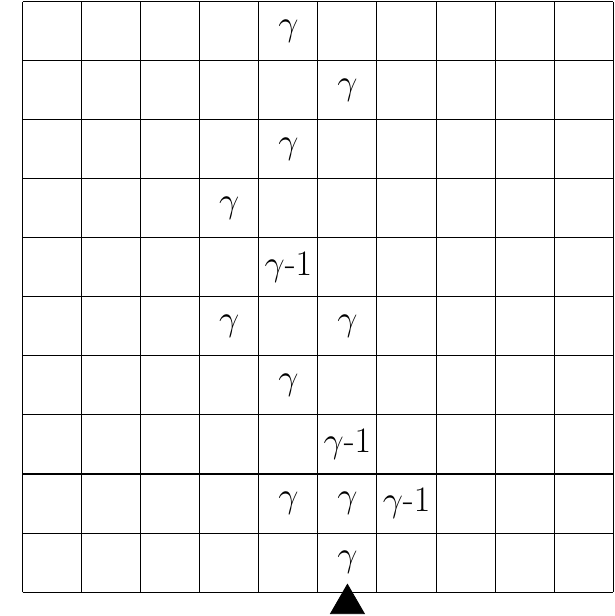}
  \caption{
    Timed crossover gates (top left),
    \emph{and} gates (top right, with $\delta=\theta_\neigh-2$),
    \emph{or} gates (top right, with $\delta=\theta_\neigh-1$),
    horizontal diode (bottom left) and
    vertical diode (bottom right)
    for neighborhoods $111$ and $127$
    with $\gamma=\theta_\neigh-1$.
    The gates have size $10\times 10$ and delay $12$.
  }
  \label{fig:timed_gates_111_127}
\end{figure}

\begin{figure}[t!]
  \centering
  \includegraphics[scale=.35]{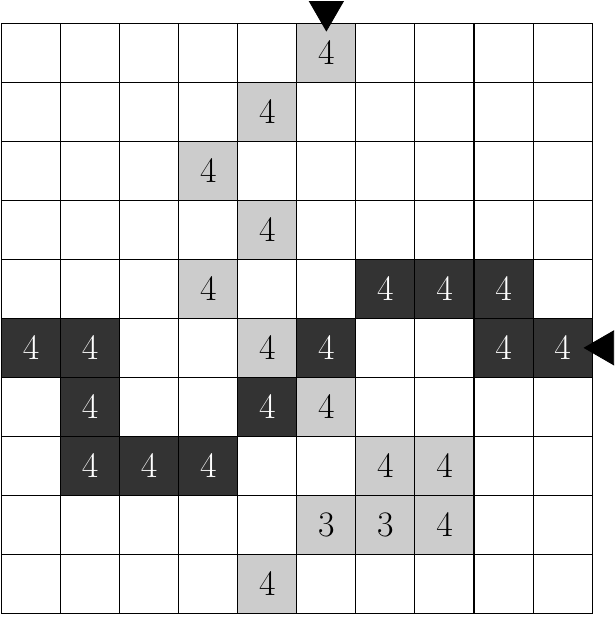}
  \hspace*{.5cm}
  \includegraphics[scale=.35]{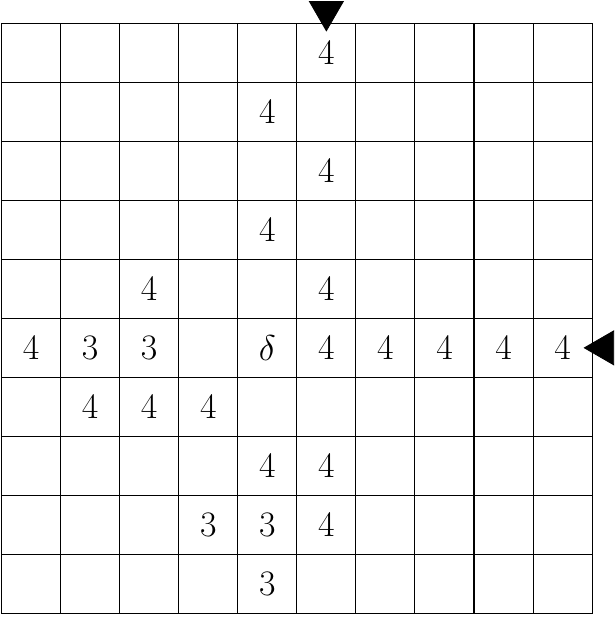}
  \\[.5em]
  \includegraphics[scale=.35]{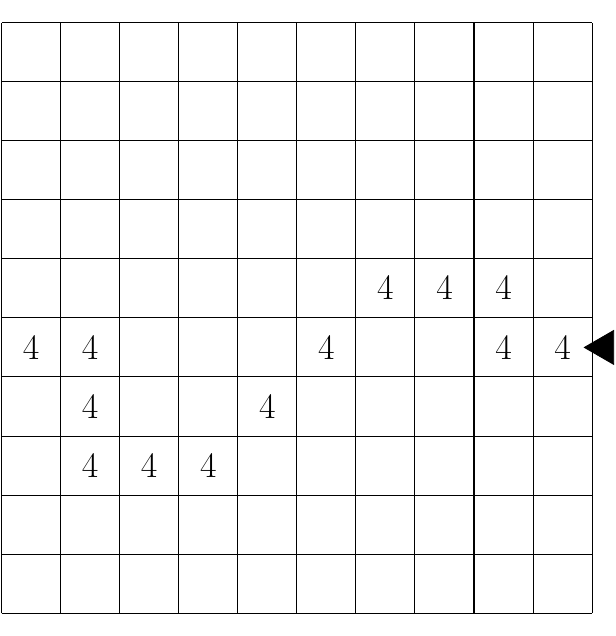}
  \hspace*{.5cm}
  \includegraphics[scale=.35]{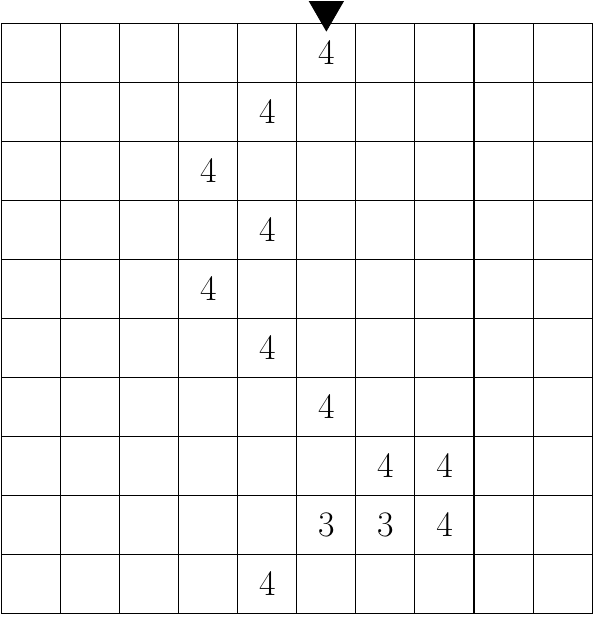}
  \caption{
    Timed crossover gate (top left),
    \emph{and} gate (top right, with $\delta=3$),
    \emph{or} gate (top right, with $\delta=4$),
    horizontal diode (bottom left) and
    horizontal diode (bottom right)
    for neighborhood $151$.
    The gates have size $10 \times 10$ and delay $13$.
  }
  \label{fig:timed_gates_151}
\end{figure}

\begin{figure}[t!]
  \centering
  \includegraphics[scale=.38]{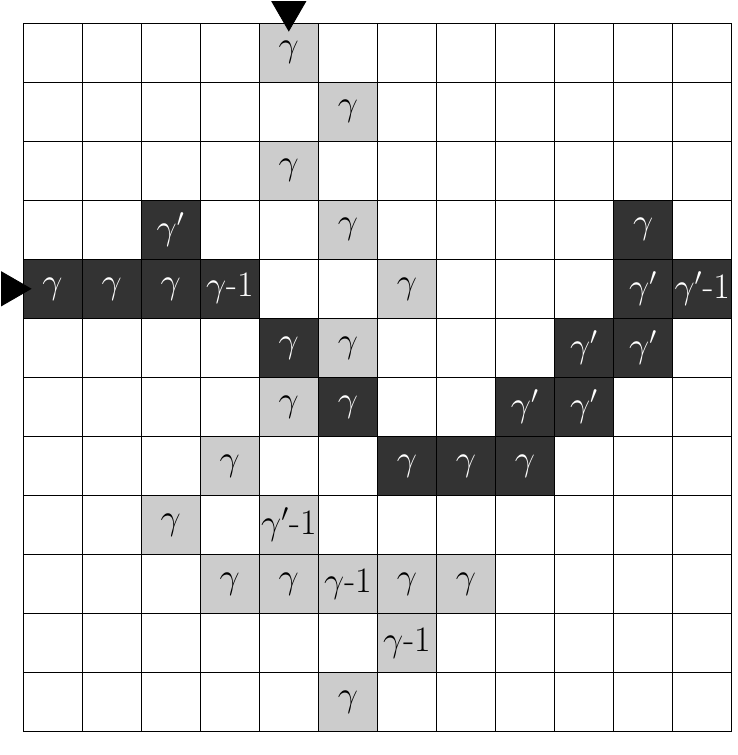}
  \hspace*{.5cm}
  \includegraphics[scale=.38]{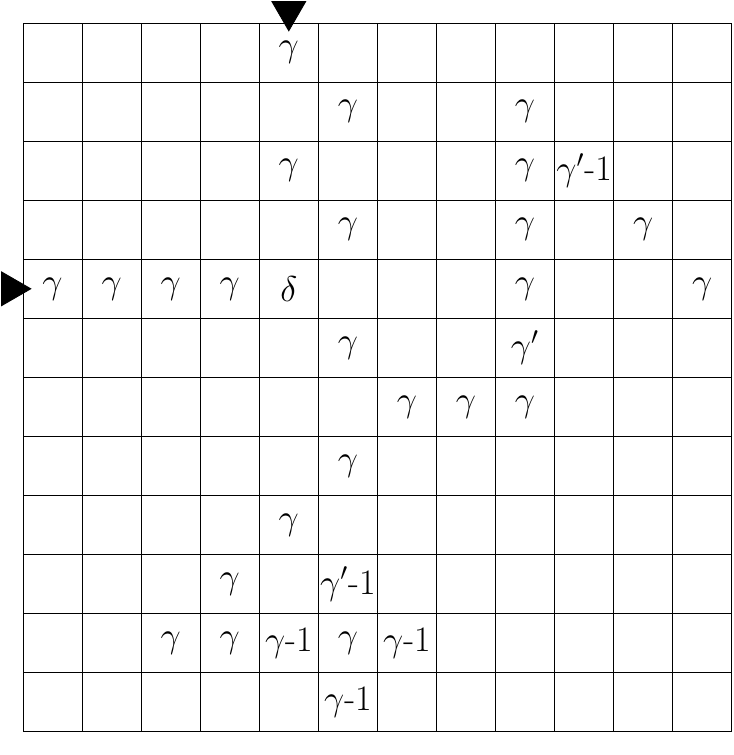}
  \\[.5em]
  \includegraphics[scale=.38]{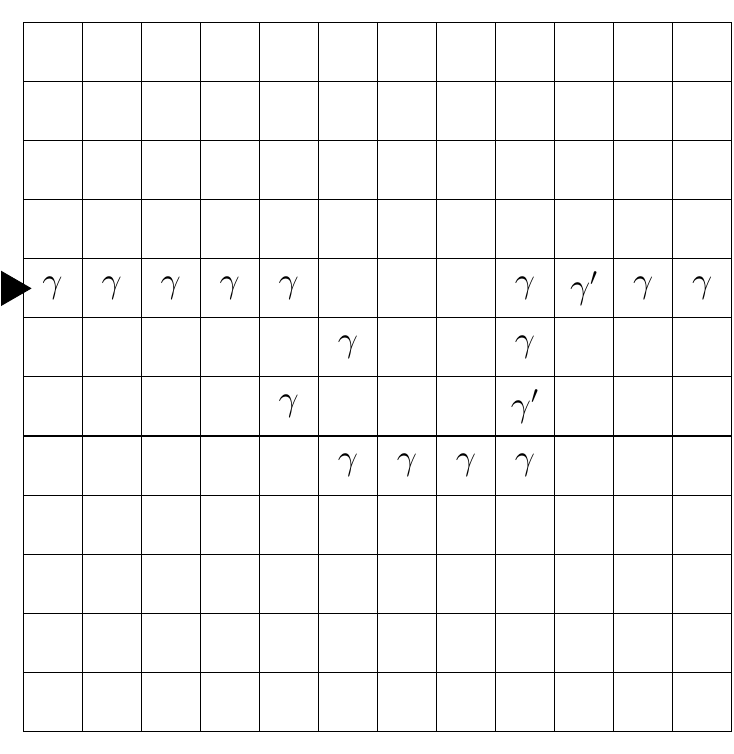}
  \hspace*{.5cm}
  \includegraphics[scale=.38]{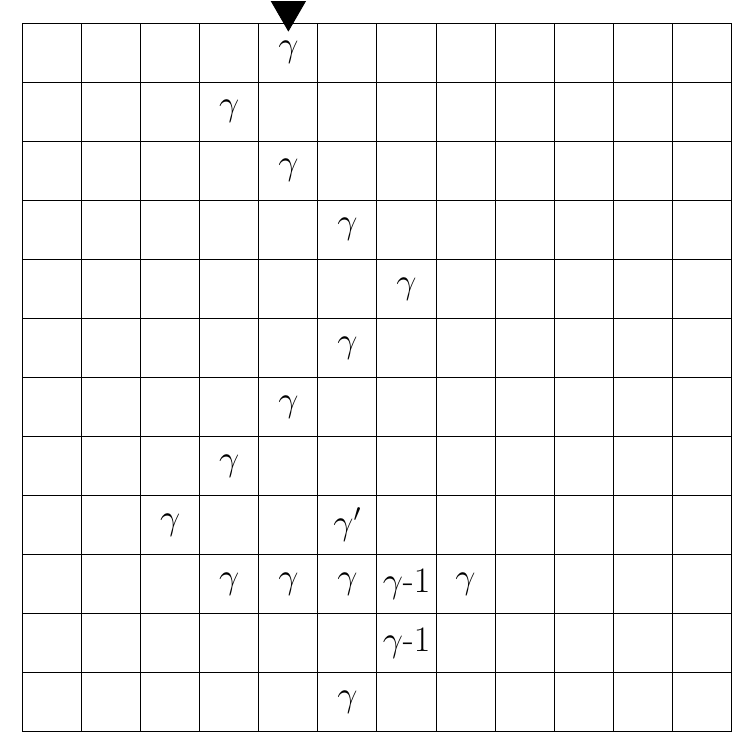}
  \caption{
    Timed crossover gates (top left),
    \emph{and} gates (top right, with $\delta=\theta_\neigh-2$),
    \emph{or} gates (top right, with $\delta=\theta_\neigh-1$),
    horizontal diode (bottom left) and
    vertical diode (bottom right)
    for neighborhoods
    $195$ (with $\gamma'=\theta_\neigh-1$),
    $199$ (with $\gamma'=\theta_\neigh-2$),
    $211$ (with $\gamma'=\theta_\neigh-1$) and
    $215$ (with $\gamma'=\theta_\neigh-2$),
    with $\gamma=\theta_\neigh-1$.
    The gates have size $12\times 12$ and delay $17$.
  }
  \label{fig:timed_gates_195_211_215_199}
\end{figure}

  The set of gates are presented as follows:
  \begin{itemize}
    \item on Figure~\ref{fig:timed_gates_111_127} for neighborhoods $111$ and $127$,
    \item on Figure~\ref{fig:timed_gates_151} for neighborhood $151$,
    \item on Figure~\ref{fig:timed_gates_195_211_215_199} for neighborhoods $195$, $211$, $215$ and $199$.
  \end{itemize}
  For each neighborhood we present a crossover gate, an \emph{and} gates, an \emph{or} gate, and diodes.
  The constant $1$ gate, turns and signal-duplication are derived from the \emph{or} gate.
  The constant $0$ gate is empty.
  \qed
\end{proof}


\subsection{Timed crossover possibility: delay issue}
\label{ss:crossover-T-delay}

Neighborhoods studied in this subsection:

\begin{center}
  \mooregridcenternumber{00100111} \mooregridcenternumber{01000011}
  \mooregridcenternumber{01010011} \mooregridcenternumber{01010111}
  \mooregridcenternumber{01011111} \mooregridcenternumber{10000011}
\end{center}

For multiple neighborhoods, we are able to construct a timed crossover gate,
and even \emph{and} and \emph{or} gates,
but we are faced with an issue when trying to plug these gates together.
Timed crossover gates for neighborhoods $39$, $67$, $83$, $87$ and $95$
are presented on Figure~\ref{fig:timed_crossovers},
and a timed crossover gate for neighborhood $131$ is presented on Figure~\ref{fig:timedcrossovergate}.

\begin{figure}[t!]
  \centering
  \begin{subfigure}[t]{.55\textwidth}
    \centering
    \includegraphics[scale = .5]{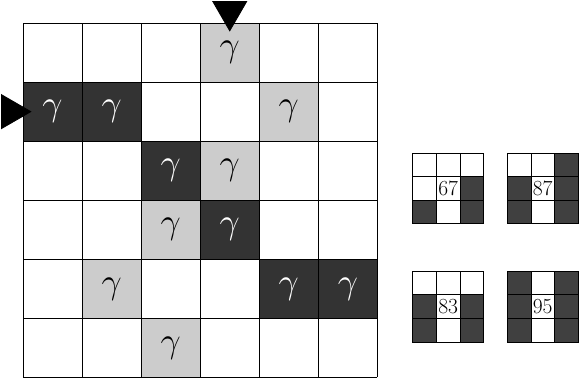}
    \caption{Neighborhoods $67$, $83$, $87$, and $95$}
  \end{subfigure}
  \begin{subfigure}[t]{.38\textwidth}
    \centering
    \includegraphics[scale = .5]{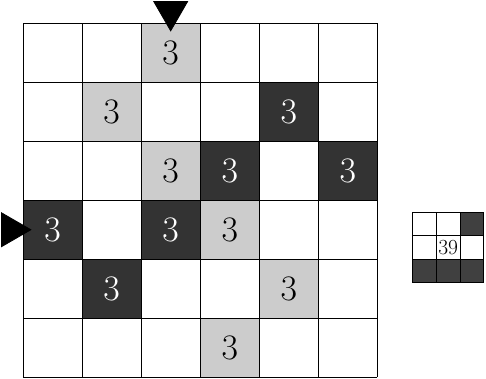}
    \caption{Neighborhoods $39$}
  \end{subfigure}
  \caption{Neighborhoods next to their timed crossover gate. Each non-empty cell contains $\gamma=\theta_\neigh-1$ grains. The triangles indicate the direction in which the cells are toppled.}
  \label{fig:timed_crossovers}
\end{figure}

Let us focus on neighborhood $39$ as an example.
On Figure~\ref{fig:timed_39} we present an almost complete set of gates, of size $11\times 11$ with delay $11$.
Inputs are located on the \emph{west} and \emph{north} borders,
and outputs on the \emph{east} and \emph{south} borders.
The north input and south output are identicaly located in every gate,
and there are two possibles locations of the west input and east output (size $11$ is odd):
either exactly in the middle of their respective borders (label $A$ on Figure~\ref{fig:timed_39}),
or one cell above (label $B$).
They are meant to be plugged alternatively to each other horizontaly
(adapting the construction from Figure~\ref{fig:cvp-grid}),
following a pattern of mono-labeled columns of gates.
There are:
\begin{itemize}
  \item \emph{wires} (also \emph{diodes}) labeled $A$ and $B$,
  \item a \emph{crossover} gate labeled $A$,
  \item an \emph{or} gate labeled $B$,
  \item an \emph{and} gate labeled $B$,
  \item a \emph{constant} $1$ gate labeled $B$ is derived from the \emph{or} gate,
  \item a \emph{signal-duplication} gate labeled $B$ is derived from the \emph{or} gate,
  \item \emph{turn} gates labeled $B$ are derived from the \emph{or} gate,
  \item a \emph{constant} $0$ gate is empty (labeled $A$ and $B$).
\end{itemize}
Only a pair of gates is missing: \emph{turn} gates labeled $A$.
If turn gates labeled $A$ exist,
then we can connect all these gates together following this mono-labeled column approach
(using zigzags to switch label)
and prove the $\Poly$-completeness of \textbf{$\neigh$-TIMED-PRED}.
However, we show below that turn gates labeled $A$ do not exist:
the distance from the west input corresponding to label $A$ to the south output
is incompatible with the movements possible under a delay of $11$.
We do not know how to proceed with this reduction without turn gates labeled $A$.

\begin{figure}
  \centering
  \includegraphics[scale=.4]{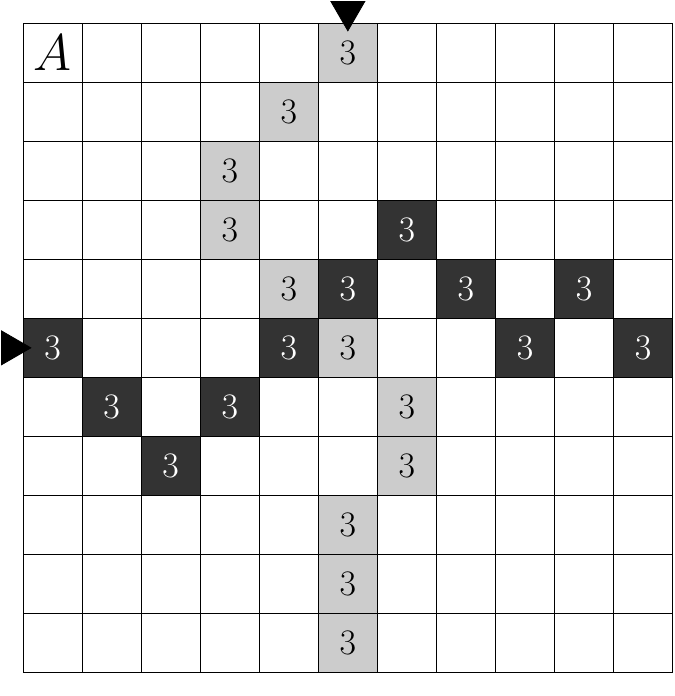}
  \hspace*{.5cm}
  \includegraphics[scale=.4]{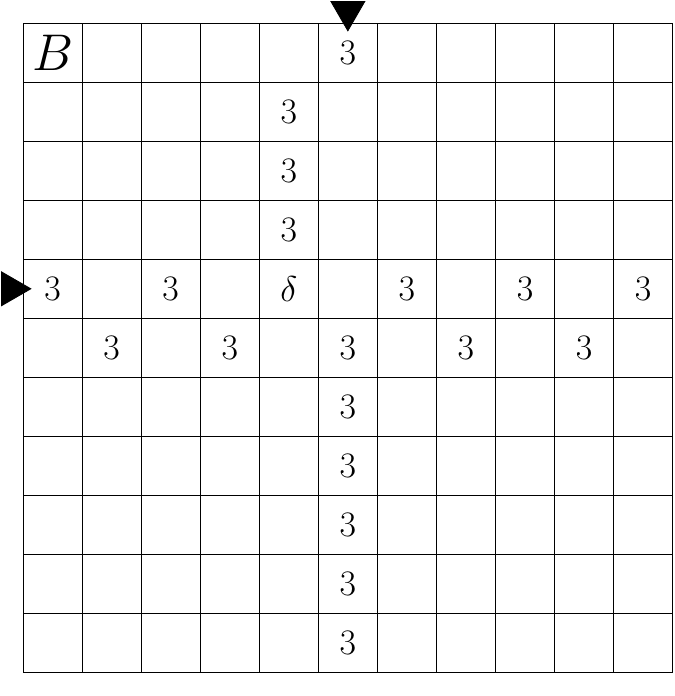}
  \\[.5em]
  \includegraphics[scale=.3]{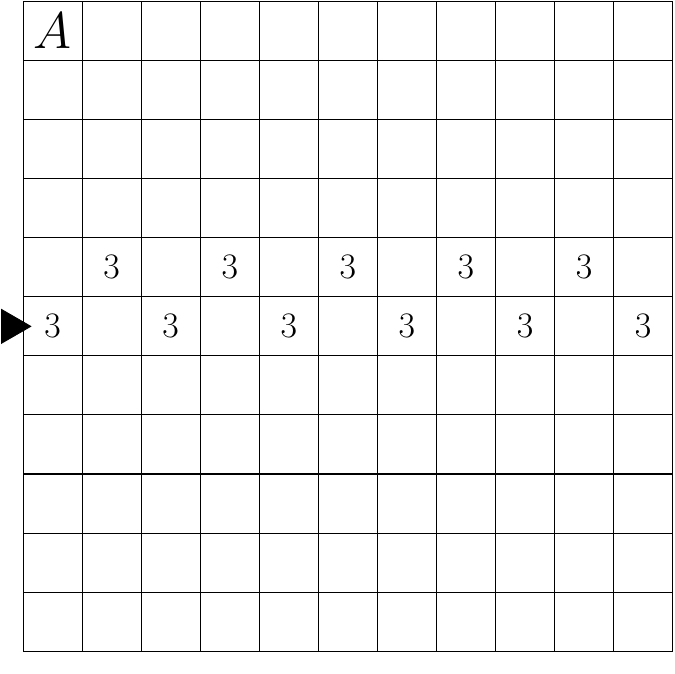}
  \hspace*{.5cm}
  \includegraphics[scale=.3]{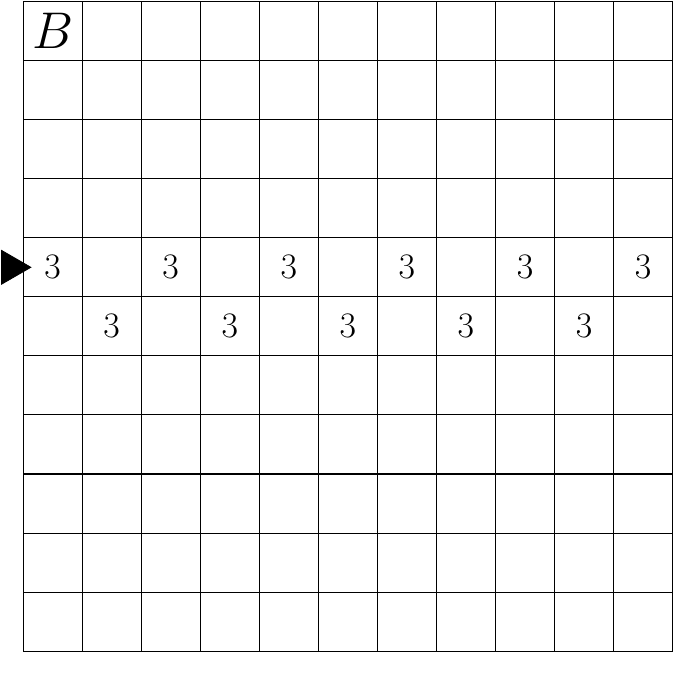}
  \hspace*{.5cm}
  \includegraphics[scale=.3]{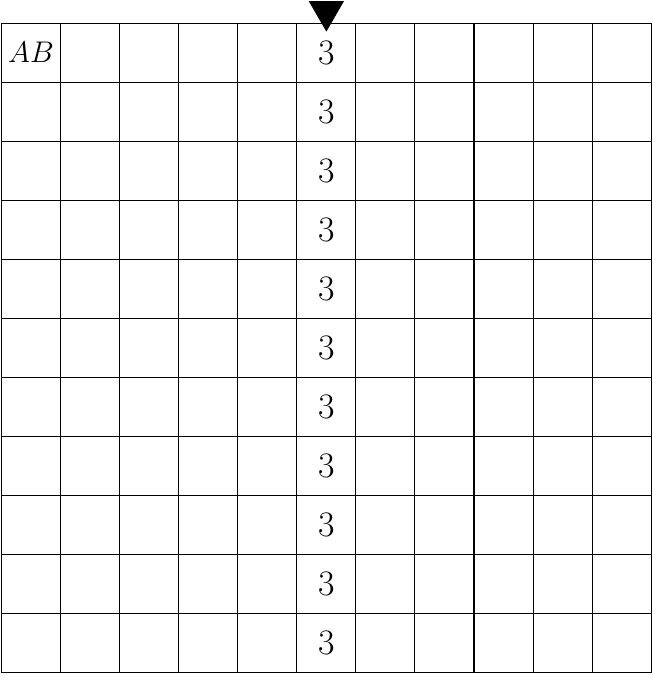}
  \caption{
    An almost complete set of timed gates of size $11\times 11$ and delay $11$
    for neighborhood $39$.
    Timed crossover gate (top left, labeled $A$),
    \emph{or} gate (top right, with $\delta=3$, labeled $B$),
    \emph{and} gate (top right, with $\delta=2$, labeled $B$),
    horizontal diodes  (bottom left and center, resp. labeled $A$ and $B$)
    and vertical diode (bottom right, labeled both $A$ and $B$).
  }
  \label{fig:timed_39}
\end{figure}

The synchronization issue is explained by the fact that neighborhoods
$39$, $67$, $83$, $87$, $95$ and $131$ are all in one of the two following cases:
\begin{itemize}
  \item not having cell \emph{north} nor \emph{south} in the neighborhood,
  \item not having cell \emph{west} nor \emph{east} in the neighborhood.
\end{itemize}
Consequently, at each step of the avalanche process (chain of reaction of topplings),
a signal must move along the perpendicular direction.
For simplicity let us assume that neither cell \emph{west} nor \emph{east} are in the neighborhood,
so that a signal must move along the $y$ axis at each step.
This induces a parity for the timestamps of signals:
they have even timestamps on the cells $X_2=\{(x,y)\in\Z^2 \mid x \text{ is even}\}$
and odd timestamps on the other cells $\Z^2\setminus X_2$,
or odd timestamps on $X_2$ and even timestamps on $\Z^2\setminus X_2$.

It turns out that the \emph{or/and} and \emph{crossover} gates presented on
Figure~\ref{fig:timed_39} use a different combination of parities for the two input signals, because:
\begin{itemize}
  \item \emph{or/and} gates require the two signals to \emph{meet} at synchronized timestamps,
  \item \emph{crossover} gates require the two signals to \emph{be adjacent along the $y$ axis} at synchronized timestamps.
\end{itemize}
This leads to an incompatiblity in the synchronization between those gates for neighborhood $39$.
Neighborhoods $67$, $83$, $87$ and $95$ face a symmetric issue.

The case of $131$ is different. It does not have the cell \emph{west} nor \emph{east}
in the neighborhood, therefore it has a parity along the $y$ axis,
but the \emph{crossover} gate requires the two signals to be adjacent along the $x$ axis.
So there is no apparent incompatibility.
Nonetheless, we were not able to solve the synchronization issue for neighborhood $131$.

\subsection{Timed crossover impossibility: conjectured}
\label{ss:crossover-T-impossibility}

Neighborhoods studied in this subsection:

\begin{center}
  \mooregridcenternumber{00100011} \hspace*{-.67cm}
  \mooregridcenternumber{01100011} \hspace*{-.67cm}
  \mooregridcenternumber{01100111} \hspace*{-.67cm}
  \mooregridcenternumber{01110011} \hspace*{-.67cm}
  \mooregridcenternumber{01110111} \hspace*{-.67cm}
  \mooregridcenternumber{10100011} \hspace*{-.67cm}
  \mooregridcenternumber{11100011} \hspace*{-.67cm}
  \mooregridcenternumber{11110011} \hspace*{-.67cm}
  \mooregridcenternumber{11110111} \hspace*{-.67cm}
  \mooregridcenternumber{11111111}
\end{center}

The difficulty to design a timed crossover gate for these neighborhoods
lies in the fact that it is not possible to take advantage of the diagonal
neighbors (in terms of delay) because the cell in-between the two diagonals
is also an out-neighbor. We conjecture that this prevents to construct timed
crossover gates for these neighborhoods, \emph{i.e.}~that they do not admit a timed crossover gate.

\section{Conclusions}

In this work we have considered the timed version of the sandpile prediction problem,
asking whether, following a single grain addition on a stable configuration,
a questionned cell will topple or not at a precise time step $t\in\N$ which is part
of the input.
To study this problem, we have introduced the notion of timed firing graph,
and timed crossover gate. The timed crossover gate allows for each of the two input cells
to trigger an avalanche eventually toppling both output cells,
provided there is a strictly positive retard on the output signal in the incorrect direction
(not correspond to a crossing of two signal, \emph{e.g.}~from
\emph{north} to \emph{south} the avalanche is allowed to topple the \emph{east} cell
if the delay of the signal from \emph{north} to \emph{south}
is striclty smaller than the delay of the signal from \emph{north} to \emph{east}).
This difference of delay aims at being embedded in a reduction from \textbf{MCVP}
to \textbf{$\neigh$-TIMED-PRED} problem, according to the construction depicted on Figure~\ref{fig:cvp-grid}
(which has been designed with this purpose).
It requires to have \emph{wire}, \emph{turns}, \emph{and} and \emph{or} gates with adequate delays,
and with input and output cells that plug well to each other in order to keep the signals synchronized.
Our results are partial; among the $255$ non-empty subsets of the Moore neighborhood
(the results are summarized on Table~\ref{tab:res_cross-T}):
\begin{itemize}
  \item $21$ do not span $\Z^2$ and are therefore not relevant to our study,
  \item $52$ 
    admit a complete timed gates toolkit (including a timed crossover gate) and therefore
    have a $\Poly$-complete timed prediction problem \textbf{$\neigh$-TIMED-PRED},
  \item $99$ 
    have a planar sandpile structure and therefore do not admit a timed crossover gate,
  \item $34$ 
    admit a timed crossover gate, but one is faced with an issue in the delays when connecting the gates together,

  \item the $49$ 
    remaining neighborhoods (including Moore) are conjectured to
    do not admit a timed crossover gate.
\end{itemize}
Remark that a crossover gate is also a timed crossover gate,
hence the $12$ neighborhoods having a $\Poly$-complete \textbf{$\neigh$-PRED} problem
are part of the $52$ neighborhoods having a $\Poly$-complete \textbf{$\neigh$-TIMED-PRED} problem.
It is notable that $40$ neighborhoods (from $7$ equivalence classes)
do not admit a crossover gate, but do admit a complete timed gates toolkit
allowing to prove the $\Poly$-completeness of timed prediction problem.
Intuitively, following Banks approach, for these neighborhoods
the embedding of the circuit evaluation in the sandpile dynamics
is possible when one follows precisely the time steps,
but asymptotically the sandpile will not stay in a configuration where the result of the circuit can be read out.

The precise characterization of the neighborhoods for which it is impossible to follow Banks approach
in proving the $\Poly$-completeness of \textbf{$\neigh$-TIMED-PRED} is still open,
and does not seem to be captured by the existence of timed crossover gate as defined in Section~\ref{s:def_soa}.
Proving that timed crossover gate do not exist for von Neumann and Moore neighborhoods
would reveal more obstacles to the embedding of efficient computation within them.

\begin{credits}
\subsubsection{\ackname}
The authors thank projects
ANR-24-CE48-7504 ALARICE,
HO\-RIZON-MSCA-2022-SE-01 101131549 ACANCOS,
STIC AmSud CAMA 22-STIC-02.

\subsubsection{\discintname}
The authors have no competing interests to declare that are
relevant to the content of this article.
 \end{credits}
%
%
%
\bibliographystyle{splncs04}
\bibliography{biblio}
\end{document}

%% file: preamble.tex
\usepackage[T1]{fontenc}
\usepackage{graphicx}
\usepackage{amsmath,amssymb,amsfonts}%
\usepackage{stmaryrd}%
\usepackage[export]{adjustbox}
\usepackage{makecell}
\usepackage{bitset}%
\usepackage{tikz}%
\usepackage{ifthen}%
\usepackage{multirow}%
\usepackage{subcaption}%

\graphicspath{{figures/}}
\renewcommand{\int}[1]{[#1]}
\newcommand{\intz}[1]{\llbracket#1\rrbracket}
\newcommand{\HH}[1]{\mathtt{H}\left({#1}\right)}
\newcommand{\N}{\mathbb{N}}
\newcommand{\Z}{\mathbb{Z}}

\newcommand{\NC}{{\mathsf{NC}}}
\newcommand{\Poly}{{\mathsf{P}}}

\newcommand{\fsubset}{{\subset_{\texttt{fin}}}}
\newcommand{\neigh}{{\mathcal N}}

\newcommand{\moore}{{\mathcal N_{\texttt{m}}}}
\newcommand{\vn}{{\mathcal N_{\texttt{vn}}}}

\newcommand{\grains}[1]{{\mathcal G({#1})}}
\newcommand{\one}[1]{{\mathsf 1}}

\newcommand{\decisionpb}[4]{\medskip\noindent\fbox{\parbox{\textwidth}{
{#1} ({\bf #2})\\
{Input:} #3\\
{Question:} #4
}}\medskip}

\newcommand{\mooregrid}[2][0]{
\bitsetSetBin{moorestr}{#2}
  \begin{tikzpicture}[scale=.4]
    \ifthenelse{\equal{\bitsetGet{moorestr}{3}}{1}}{\draw[fill=black!75] (-1,1) rectangle ++(1,1);}{\draw[fill=white] (-1,1) rectangle ++(1,1);}
    \ifthenelse{\equal{\bitsetGet{moorestr}{7}}{1}}{\draw[fill=black!75] (0,1) rectangle ++(1,1);}{\draw[fill=white] (0,1) rectangle ++(1,1);}
    \ifthenelse{\equal{\bitsetGet{moorestr}{2}}{1}}{\draw[fill=black!75] (1,1) rectangle ++(1,1);}{\draw[fill=white] (1,1) rectangle ++(1,1);}
    \ifthenelse{\equal{\bitsetGet{moorestr}{6}}{1}}{\draw[fill=black!75] (1,0) rectangle ++(1,1);}{\draw[fill=white] (1,0) rectangle ++(1,1);}
    \ifthenelse{\equal{\bitsetGet{moorestr}{1}}{1}}{\draw[fill=black!75] (1,-1) rectangle ++(1,1);}{\draw[fill=white] (1,-1) rectangle ++(1,1);}
    \ifthenelse{\equal{\bitsetGet{moorestr}{5}}{1}}{\draw[fill=black!75] (0,-1) rectangle ++(1,1);}{\draw[fill=white] (0,-1) rectangle ++(1,1);}
    \ifthenelse{\equal{\bitsetGet{moorestr}{0}}{1}}{\draw[fill=black!75] (-1,-1) rectangle ++(1,1);}{\draw[fill=white] (-1,-1) rectangle ++(1,1);}
    \ifthenelse{\equal{\bitsetGet{moorestr}{4}}{1}}{\draw[fill=black!75] (-1,0) rectangle ++(1,1);}{\draw[fill=white] (-1,0) rectangle ++(1,1);}
    \node at (.5,.5) {$\times$};
    \ifthenelse{\equal{#1} {1}} {\node at (.5,-1.5) {#2 = \bitsetGetDec{moorestr} };} {}
  \end{tikzpicture}
}

\newcommand{\mooregridcenternumber}[2][0]{
\bitsetSetBin{moorestr}{#2}
  \begin{tikzpicture}[scale=.3]
    \ifthenelse{\equal{\bitsetGet{moorestr}{3}}{1}}{\draw[fill=black!75] (-1,1) rectangle ++(1,1);}{\draw[fill=white] (-1,1) rectangle ++(1,1);}
    \ifthenelse{\equal{\bitsetGet{moorestr}{7}}{1}}{\draw[fill=black!75] (0,1) rectangle ++(1,1);}{\draw[fill=white] (0,1) rectangle ++(1,1);}
    \ifthenelse{\equal{\bitsetGet{moorestr}{2}}{1}}{\draw[fill=black!75] (1,1) rectangle ++(1,1);}{\draw[fill=white] (1,1) rectangle ++(1,1);}
    \ifthenelse{\equal{\bitsetGet{moorestr}{6}}{1}}{\draw[fill=black!75] (1,0) rectangle ++(1,1);}{\draw[fill=white] (1,0) rectangle ++(1,1);}
    \ifthenelse{\equal{\bitsetGet{moorestr}{1}}{1}}{\draw[fill=black!75] (1,-1) rectangle ++(1,1);}{\draw[fill=white] (1,-1) rectangle ++(1,1);}
    \ifthenelse{\equal{\bitsetGet{moorestr}{5}}{1}}{\draw[fill=black!75] (0,-1) rectangle ++(1,1);}{\draw[fill=white] (0,-1) rectangle ++(1,1);}
    \ifthenelse{\equal{\bitsetGet{moorestr}{0}}{1}}{\draw[fill=black!75] (-1,-1) rectangle ++(1,1);}{\draw[fill=white] (-1,-1) rectangle ++(1,1);}
    \ifthenelse{\equal{\bitsetGet{moorestr}{4}}{1}}{\draw[fill=black!75] (-1,0) rectangle ++(1,1);}{\draw[fill=white] (-1,0) rectangle ++(1,1);}
    \node at (.5,.5) {\tiny\bitsetGetDec{moorestr}};
    \ifthenelse{\equal{#1}{1}}{
      \draw[black!75, very thick] (-1.1, -1.1) rectangle ++(3.2, 3.2);
      \draw[black!75, very thick, fill = black!75] (-1.1,-2) rectangle ++(3.2, .9);
    }{}
    \node[text={\ifnum#1=1 white\else black\fi}] at (.5,-1.5) {\scriptsize #2};
  \end{tikzpicture}
}


%% file: results_table.tex
\begin{tabular}{|c|c|c|c|}
    \hline
    Section & \ref{ss:crossover-T} & \ref{ss:crossover-T-delay} & \ref{ss:crossover-T-impossibility}\\
    \hline
    Type & $\Poly$-complete & Delay issue & \makecell{Timed crossover impossibility\\(conjecture)}\\
    \hline
    \rotatebox{90}{Neigborhoods} & \makecell{\\ \mooregridcenternumber{01101111} \mooregridcenternumber{01111111} \\ \mooregridcenternumber{10000111} \mooregridcenternumber{10001111} \\ \mooregridcenternumber{10010111} \mooregridcenternumber{11000011} \\
    \mooregridcenternumber{11000111} \mooregridcenternumber{11010011} \\
    \mooregridcenternumber{11010111}\\} & \makecell{\mooregridcenternumber{00100111} \mooregridcenternumber{01000011} \\ \mooregridcenternumber{01010011} \mooregridcenternumber{01010111} \\ \mooregridcenternumber{01011111} \mooregridcenternumber{10000011} \\}& \makecell{ \\ \mooregridcenternumber{00100011} \mooregridcenternumber{01100011} \\ \mooregridcenternumber{01100111} \mooregridcenternumber{01110011} \\ \mooregridcenternumber{01110111} \mooregridcenternumber{10100011} \\ \mooregridcenternumber{11100011} \mooregridcenternumber{11110011} \\ \mooregridcenternumber{11110111} \mooregridcenternumber{11111111} \\}\\
    \hline
\end{tabular}